\documentclass[copyright,creativecommons]{eptcs}
\pdfoutput=1  
\usepackage[utf8]{inputenc}

\usepackage{underscore} 

\title{Emptiness Problems for Distributed Automata}
\author{
  Antti Kuusisto
  \institute{University of Bremen\\ Germany}
  \email{antti.j.kuusisto@gmail.com}
  \and
  Fabian Reiter
  \institute{IRIF, Université Paris Diderot\\ France}
  \email{fabian.reiter@gmail.com}
}
\def\titlerunning{Emptiness Problems for Distributed Automata}
\def\authorrunning{A. Kuusisto \& F. Reiter}

\usepackage{microtype}
\usepackage[fleqn]{amsmath}
\usepackage{amssymb}
\usepackage{amsthm}
\usepackage{mathtools}
\usepackage{stmaryrd}
\usepackage{tikz}
\usetikzlibrary{matrix, positioning, decorations.pathreplacing, shapes.multipart}
\usepackage{standalone}
\usepackage{xcolor}
\usepackage{xifthen}
\usepackage[listof=false,labelsep=period,font={small},labelfont={bf}]{caption}

\usepackage{sty/commands}

\theoremstyle{plain}
\newtheorem{theorem}{Theorem}
\newtheorem{proposition}[theorem]{Proposition}
\newtheorem{lemma}[theorem]{Lemma}
\theoremstyle{definition}
\newtheorem{definition}[theorem]{Definition}

\begin{document}

\maketitle

\begin{abstract}
  We investigate the decidability of
  the emptiness problem for three classes of distributed automata.
  These devices operate on finite directed graphs,
  acting as networks of identical finite-state machines
  that communicate in an infinite sequence of synchronous rounds.
  The problem is shown to be decidable in $\LogSpace$
  for a class of forgetful automata,
  where the nodes see the messages received from their neighbors
  but cannot remember their own state.
  When restricted to the appropriate families of graphs,
  these forgetful automata are
  equivalent to classical finite word automata,
  but strictly more expressive than finite tree automata.
  On the other hand,
  we also show that the emptiness problem is undecidable in general.
  This already holds
  for two heavily restricted classes of distributed automata:
  those that reject immediately
  if they receive more than one message per round,
  and those whose state diagram must be acyclic except for self-loops.
\end{abstract}

\section{Introduction}
Recent years have seen increased interest in
automata theoretic approaches to the study of
distributed message-passing algorithms.
Such algorithms are executed concurrently
by all nodes of an arbitrary computer network
in order to solve some graph problem
related to the network structure.
The weakest classes of these algorithms
can be represented as deterministic finite-state machines,
here referred to as \emph{distributed automata},
which run as follows on a finite labeled directed graph:
We place a copy of the same machine on every node of the graph
and let the nodes communicate
in an infinite sequence of synchronous rounds.
In every round,
each node computes its next local state
as a function of its own current state
and the set of current states of its incoming neighbors.
(The states of the incoming neighbors represent
incoming messages sent by the neighbors.)
Acting as a semi-decider,
the machine at a given node accepts
precisely if it visits an accepting state at some point in time.

In a recently initiated research program,
several classes of distributed algorithms
have been given logical characterizations
in the spirit of descriptive complexity theory
\cite{DBLP:books/daglib/Immerman99},
and conversely,
some well-known logics have been provided with novel
machine-oriented characterizations.
First,
in~\cite{DBLP:conf/podc/HellaJKLLLSV12,DBLP:journals/dc/HellaJKLLLSV15},
Hella~et~al.\ established the equivalence of
\emph{local} distributed automata and basic \emph{modal logic};
in the context of distributed computing,
the term “local” means
that nodes stop changing their state after a constant number of rounds
(see, e.g., \cite{DBLP:journals/csur/Suomela13}).
The link with logic was further strengthened by Kuusisto
in~\cite{DBLP:conf/csl/Kuusisto13},
where a logical characterization
for unrestricted (nonlocal) automata was obtained
in terms of a modal-logic-based variant of
Datalog called \emph{modal substitution calculus} (MSC).
Then,
in~\cite{DBLP:conf/lics/Reiter15},
Reiter extended local distributed automata
with a global acceptance condition
and the ability to alternate between
nondeterministic and parallel computations,
thereby providing an automata-theoretic
characterization of \emph{monadic second-order logic} ($\MSO$)
on arbitrary graphs.
Similarly,
the least fixpoint fragment of the modal $\mu$-calculus
has been characterized in~\cite{DBLP:conf/icalp/Reiter17}
using
an asynchronous subclass of nonlocal distributed automata.
Furthermore,
the descriptive complexity approach
of~\cite{DBLP:conf/podc/HellaJKLLLSV12,DBLP:journals/dc/HellaJKLLLSV15}
and~\cite{DBLP:conf/csl/Kuusisto13}
found an application in~\cite{DBLP:journals/corr/Kuusisto14a},
where tools from logic were used
to show that
universally halting distributed automata are necessarily local
if we allow infinite networks into the picture.

As the above equivalences are all effective,
we can immediately settle the decidability question
of the emptiness problem for local automata:
it is decidable for the basic variant
of~\cite{DBLP:conf/podc/HellaJKLLLSV12,DBLP:journals/dc/HellaJKLLLSV15},
but undecidable for the extension considered
in~\cite{DBLP:conf/lics/Reiter15}.
This is because the (finite) satisfiability problem
is $\PSpace$-complete for basic modal logic
but undecidable for $\MSO$.
The problem is also decidable for the asynchronous class
of~\cite{DBLP:conf/icalp/Reiter17},
since (finite) satisfiability for the $\mu$-calculus
is $\ExpTime$-complete.
However,
the corresponding question for unrestricted automata
was left open in~\cite{DBLP:conf/csl/Kuusisto13}.
In the present paper,
we answer this question negatively for the general case
and also consider it for three subclasses of distributed automata.

Our first variant, dubbed \emph{forgetful} automata,
is characterized by the fact
that nodes can see their incoming neighbors' states
but cannot remember their own state.
Although this restriction might seem very artificial,
it bears an intriguing connection to classical automata theory:
forgetful distributed automata turn out to be
equivalent to finite word automata
(and hence $\MSO$)
when restricted to directed paths,
but strictly more expressive than finite tree automata
(and hence $\MSO$)
when restricted to ordered directed trees.
As pointed out in~\mbox{\cite[Prp.~8]{DBLP:conf/csl/Kuusisto13}},
the situation is different on arbitrary directed graphs,
where distributed automata (and hence forgetful ones)
are unable to recognize non-reachability properties
that can be easily expressed in $\MSO$.
Hence,
none of the two formalisms can simulate the other in general.
However,
while satisfiability for $\MSO$ is undecidable,
we obtain a $\LogSpace$ algorithm
that decides the emptiness problem for forgetful distributed automata.

The preceding decidability result begs the question of
what happens if we drop the forgetfulness condition.
Motivated by the equivalence of
finite word automata and
forgetful distributed automata on paths,
we first investigate this question
when restricted to directed paths.
In sharp contrast to the forgetful case,
we find that
for arbitrary distributed automata,
it is undecidable
whether an automaton accepts on some directed path.
Although our proof follows the standard approach
of simulating a Turing machine,
it has an unusual twist:
we exchange the roles of space and time,
in the sense that
the space of the simulated Turing machine $\Machine$
is encoded into
the time of the simulating distributed automaton $\Automaton$,
and conversely,
the time of $\Machine$ is encoded into the space of $\Automaton$.
To lift this result to arbitrary graphs,
we introduce the class of \emph{monovisioned} distributed automata,
where nodes enter a rejecting sink state as soon as
they see more than one state in their incoming neighborhood.
For every distributed automaton~$\Automaton$,
one can construct a monovisioned automaton~$\Automaton'$
that satisfies the emptiness property
if and only if $\Automaton$ does so on directed paths.
Hence,
the emptiness problem is undecidable for monovisioned automata,
and thus also in general.

Our third and last class consists of those distributed automata
whose state diagram does not contain any directed cycles,
except for self-loops;
we call them \emph{quasi-acyclic}.
The motivation for this particular class is threefold.
First,
quasi-acyclicity may be seen as a natural intermediate stage
between local and unrestricted distributed automata,
because local automata
(for which the emptiness problem is decidable)
can be characterized as those automata
whose state diagram is acyclic as long as we ignore sink states
(i.e., states that cannot be left once reached).
Second,
the Turing machine simulation mentioned above
makes crucial use of directed cycles
in the diagram of the simulating automaton,
which suggests that cycles might be the source of undecidability.
Third,
the notion of quasi-acyclic state diagrams also plays a major role
in~\cite{DBLP:conf/icalp/Reiter17},
where it serves as an ingredient
for the aforementioned subclass of asynchronous distributed automata
(for which the emptiness problem is also decidable).
However,
contrary to what one might expect from these clues,
we show that quasi-acyclicity alone is not sufficient
to make the emptiness problem decidable,
thereby giving an alternative proof
of undecidability for the general case.

The remainder of this paper is organized as follows:
We first introduce the formal definitions in Section~\ref{sec:preliminaries}
and establish the connections between forgetful distributed automata
and classical word and tree automata in Section~\ref{sec:classical}.
Then,
we show the positive decidability result for forgetful automata
in Section~\ref{sec:forgetfulness}.
Finally,
we establish the negative results for monovisioned automata
in Section~\ref{sec:space-time}
and for quasi-acyclic automata in Section~\ref{sec:fireworks}.
\section{Preliminaries}
\label{sec:preliminaries}
We denote the set of non-negative integers by
$\defd{\Natural} = \set{0,1,2,\dots}$
and the power set of any set $S$ by \defd{$\powerset{S}$}.

Let $\Alphabet$ be a finite set of symbols
and $\RelCount$ be a positive integer.
A (finite)
\defd{$\Alphabet$-labeled, \mbox{$\RelCount$-relational} directed graph},
abbreviated \defd{\digraph{}},
is a structure
$\Graph = \tuple{\NodeSet, \tuple{\EdgeSet_k}_{1\leq k\leq\RelCount}, \Labeling}$,
where
$\NodeSet$
is a finite nonempty set of nodes,
each $\EdgeSet_k \subseteq \NodeSet \times \NodeSet$
is a set of directed edges, and
$\Labeling \colon \NodeSet \to \Alphabet$
is a labeling that assigns a symbol of $\Alphabet$ to each node.
Isomorphic \digraph{s} are considered to be equal.
If $\Node$ is a node in $\NodeSet$,
we call the pair~$\tuple{\Graph,\Node}$
a \defd{pointed \digraph{}}
with \defd{distinguished node}~$\Node$.
Furthermore,
if $\Node[1]\Node[2]$ is an edge in $\EdgeSet_k$,
then $\Node[1]$ is called an \defd{incoming $k$-neighbor} of~$\Node[2]$,
or simply an \defd{incoming neighbor}.

A \defd{directed rooted tree}, or \defd{\ditree{}},
is a \digraph{}
$\Graph = \tuple{\NodeSet, \tuple{\EdgeSet_k}_{1\leq k\leq\RelCount}, \Labeling}$
that has a distinct node~$\Root$, called the \defd{root},
such that from each node~$\Node$ in~$\NodeSet$,
there is exactly one way to reach~$\Root$
by following the directed edges in $\bigcup_{1\leq k\leq\RelCount} \EdgeSet_k$,
where $\EdgeSet_i \cap \EdgeSet_j = \emptyset$ for $i \neq j$.
A \defd{pointed \ditree{}} is a pointed \digraph{} $\tuple{\Graph,\Root}$
that is composed of a \ditree{} and its root.
Moreover,
an $\RelCount$-relational \ditree{} is called \defd{ordered}
if for $1\leq k\leq\RelCount$,
every node has at most one incoming $k$-neighbor
and every node that has an incoming $(k+1)$-neighbor
also has an incoming $k$-neighbor.
As a special case,
an ordered $1$-relational \ditree{}
is referred to as a \defd{directed path}, or \defd{\dipath{}}.

We now give a general definition of distributed automata
that subsumes all the variants considered in this paper.
Simply put,
a distributed automaton is a deterministic finite-state machine
that reads sets of states instead of the usual alphabetic symbols.
To run such an automaton on a \digraph{},
we place a copy of the same machine on every node of the \digraph{}
and let the nodes communicate
in an infinite sequence of synchronous rounds.
In every round,
each node computes its next local state
as a function of its own current state
and the set of current states of its incoming neighbors.
In order to draw the comparison with
classical word and tree automata in Section~\ref{sec:classical},
we let our distributed automata operate
on labeled, multi-relational \digraph{s}.
Furthermore,
we let the nodes of those \digraph{s}
read their own label in each communication round,
as this will facilitate the definition of forgetful automata.
Whenever possible,
the rather cumbersome notation will later be simplified.

\begin{definition}[Distributed Automaton]
  A \defd{distributed automaton}
  over $\Alphabet$-labeled, $\RelCount$-rela\-tio\-nal \digraph{s}
  is a tuple
  $\Automaton = \tuple{\StateSet,\InitState,\tuple{\TransFunc_\Symbol}_{\Symbol\in\Alphabet},\AcceptSet}$,
  where
  $\StateSet$
  is a finite nonempty set of states,
  $\InitState \in \StateSet$
  is an initial state,
  $\TransFunc_\Symbol \colon \StateSet \times (\powerset{\StateSet})^\RelCount \to \StateSet$
  is a (local) transition function associated with label~$\Symbol \in \Alphabet$, and
  $\AcceptSet \subseteq \StateSet$
  is a set of accepting states.
\end{definition}

Let
$\Graph = \tuple{\NodeSet, \tuple{\EdgeSet_k}_{1\leq k\leq\RelCount}, \Labeling}$
be a $\Alphabet$-labeled, $\RelCount$-relational \digraph{}.
The \defd{run} of $\Automaton$ on $\Graph$ is an infinite sequence
$\Run = \tuple{\Run_0, \Run_1, \Run_2, \dots}$
of maps
$\Run_\Time \colon \NodeSet \to \StateSet$,
called \defd{configurations},
which are defined inductively as follows,
for $\Time \in \Natural$ and $\Node[2] \in \NodeSet$:
\begin{equation*}
  \Run_0(\Node[2]) = \InitState
  \qquad\text{and}\qquad
  \Run_{\Time+1}(\Node[2]) =
  \TransFunc_{\Labeling(\Node[2])}
      \Bigl(\Run_\Time(\Node[2]),\,
            \bigtuple{\setbuilder{\Run_\Time(\Node[1])}{\Node[1]\Node[2] \in \EdgeSet_k}
                     }_{1\leq k\leq\RelCount}
      \Bigr).
\end{equation*}
For $\Node \in \NodeSet$,
the automaton~$\Automaton$ \defd{accepts}
the pointed \digraph{} $\tuple{\Graph,\Node}$
if $\Node$ visits an accepting state at some point
in the run $\Run$ of $\Automaton$ on $\Graph$,
i.e., if there exists $\Time \in \Natural$
such that $\Run_\Time(\Node) \in \AcceptSet$.
The \defd{language} of~$\Automaton$ (or language recognized by~$\Automaton$)
is the set of all pointed \digraph{s} that $\Automaton$ accepts.

A distributed automaton is called \defd{forgetful}
if in each round,
the nodes can see their neighbors' states but cannot remember their own state.
Formally,
for
$\Automaton = \tuple{\StateSet,\InitState,\tuple{\TransFunc_\Symbol}_{\Symbol\in\Alphabet},\AcceptSet}$,
being forgetful means that
$\TransFunc_\Symbol(\State,\vec{\NeighborSet}) = \TransFunc_\Symbol(\State',\vec{\NeighborSet})$
for all $\Symbol \in \Alphabet$,\, $\State,\State' \in \StateSet$
and $\vec{\NeighborSet} \in (\powerset{\StateSet})^\RelCount$.
Therefore,
we can represent the transition functions of such an automaton as
$\TransFunc_\Symbol \colon (\powerset{\StateSet})^\RelCount \to \StateSet$.

On the other hand,
when we consider automata that are not forgetful,
we will simplify them to have a single transition function.
Instead of letting the nodes read their own label~$\Symbol$ and
choose the appropriate function $\TransFunc_\Symbol$ in each round,
we can force them to store the label in their local state and
combine all the transition functions into a single one.
Notation can be further lightened
by limiting ourselves to $1$-relational \digraph{s}.
Hence,
we shall sometimes regard a distributed automaton as a tuple
$\Automaton = \tuple{\StateSet,\InitFunc,\TransFunc,\AcceptSet}$,
where
$\InitFunc \colon \Alphabet \to \StateSet$
is an initialization function,
$\TransFunc \colon \StateSet \times \powerset{\StateSet} \to \StateSet$
is a transition function,
and $\StateSet$ and $\AcceptSet$ are as before.
The semantics is the obvious one:
each node~$\Node$ is initialized to $\InitFunc(\Labeling(\Node))$,
computes its next state by evaluating $\TransFunc$ on
its current state and the set of states of its incoming neighbors,
and accepts if at some point in time it visits a state in $\AcceptSet$.

The central concern of this paper is
the (general) \defd{emptiness problem} for several classes of distributed automata.
Given an automaton~$\Automaton$,
the problem is to decide effectively
whether the language of~$\Automaton$ is nonempty,
i.e., whether there is a pointed \digraph{} $\tuple{\Graph,\Node}$
that is accepted by~$\Automaton$.
Similarly,
the \defd{\dipath{}-emptiness problem}
is to decide if $\Automaton$ accepts some pointed \dipath{}.
\section{Comparison with classical automata}
\label{sec:classical}
The purpose of this section is to motivate
our interest in forgetful distributed automata
by establishing their connection
with classical word and tree automata.

\begin{proposition}
  \label{prp:equivalence-word-automata}
  When restricted to the class of pointed \dipath{s},
  forgetful distributed automata are
  equivalent to finite word automata
  (and thus to monadic second-order logic).
\end{proposition}
\begin{proof}
  Let us denote a (deterministic) finite word automaton
  over some finite alphabet~$\Alphabet$
  by a tuple
  $\ClAutomaton = \tuple{\ClStateSet, \ClInitState, \ClTransFunc, \ClAcceptSet}$,
  where
  $\ClStateSet$ is the set of states,
  $\ClInitState$ is the initial state,
  $\ClTransFunc \colon \ClStateSet \times \Alphabet \to \ClStateSet$
  is the transition function, and
  $\ClAcceptSet$ is the set of accepting states.

  Given such a word automaton~$\ClAutomaton$,
  we construct a forgetful distributed automaton
  ${\Automaton = \tuple{\StateSet,\InitState,\tuple{\TransFunc_\Symbol}_{\Symbol\in\Alphabet},\AcceptSet}}$
  that simulates~$\ClAutomaton$ on $\Alphabet$-labeled \dipath{s}.
  For this,
  it suffices to set $\StateSet = \ClStateSet \cup \set{\Waiting}$,\,
  $\InitState = \Waiting$,\, $\AcceptSet = \ClAcceptSet$, and
  \begin{equation*}
    \TransFunc_\Symbol(\NeighborSet) =
    \begin{cases*}
      \ClTransFunc(\ClInitState, \Symbol) & if $\NeighborSet = \emptyset$, \\
      \ClTransFunc(\ClState, \Symbol)     & if $\NeighborSet = \set{\ClState}$
                                            for some $\ClState \in \ClStateSet$, \\
      \Waiting                            & otherwise.
    \end{cases*}
  \end{equation*}
  When $\Automaton$ is run on a \dipath{},
  each node~$\Node$ starts in a waiting phase, represented by~$\Waiting$,
  and remains idle until its predecessor
  has computed the state~$\ClState$
  that $\ClAutomaton$ would have reached
  just before reading the local symbol~$\Symbol$ of~$\Node$.
  (If there is no predecessor, $\ClState$ is set to~$\ClInitState$.)
  Then, $\Node$ switches to the state~$\ClTransFunc(\ClState, \Symbol)$
  and stays there forever.
  Consequently,
  the distinguished last node of the \dipath{} will end up
  in the state reached by~$\ClAutomaton$
  at the end of the word,
  and it accepts if and only if~$\ClAutomaton$ does.

  For the converse direction,
  we convert a given forgetful distributed automaton
  $\Automaton = \tuple{\StateSet,\InitState,\tuple{\TransFunc_\Symbol}_{\Symbol\in\Alphabet},\AcceptSet}$
  into the word automaton
  $\ClAutomaton = \tuple{\ClStateSet, \ClInitState, \ClTransFunc, \ClAcceptSet}$
  with components
  $\ClStateSet = \powerset{\StateSet}$,\,
  $\ClInitState = \emptyset$,\,
  $\ClAcceptSet = \setbuilder{\StateSubset \subseteq \StateSet}
                             {\StateSubset \cap \AcceptSet \neq \emptyset}$,
  and
  \begin{equation*}
    \ClTransFunc(\ClState, \Symbol) =
    \set{\InitState} \cup 
    \begin{cases*}
      \set{\TransFunc_\Symbol(\emptyset)}
      & if $\ClState = \ClInitState$, \\
      \setbuilder{\TransFunc_\Symbol(\set{\State})}{\State \in \ClState}
      & otherwise.
    \end{cases*}
  \end{equation*}
  On any $\Alphabet$-labeled \dipath{}~$\Graph$,
  our construction guarantees that
  the set of states visited by~$\Automaton$ at the $i$-th node
  is equal to
  the state that~$\ClAutomaton$ reaches
  just after processing the $i$-th symbol
  of the word associated with~$\Graph$.
  We can easily verify this by induction on~$i$:
  At the first node,
  which is labeled with $\Symbol_1$,
  automaton~$\Automaton$ starts in state~$\InitState$
  and then remains forever in state~$\TransFunc_{\Symbol_1}(\emptyset)$.
  Node number~$i + 1$ also starts in~$\InitState$,
  and transitions to
  $\TransFunc_{\Symbol_{i+1}}(\set{\State_\Time^i})$ at time~$\Time + 1$,
  where $\Symbol_{i+1}$ is the node's own label
  and~$\State_\Time^i$ is the state of its predecessor at time~$\Time$.
  In agreement with this behavior,
  we know by the induction hypothesis and the definition of $\ClTransFunc$
  that the state of~$\ClAutomaton$ after reading $\Symbol_{i+1}$
  is precisely
  $\set{\InitState} \cup
   \setbuilder{\TransFunc_{\Symbol_{i+1}}(\set{\State_\Time^i})}{\Time \in \Natural}$.
  As a result,
  the final state reached by~$\ClAutomaton$ will be accepting
  if and only if
  $\Automaton$ visits some accepting state at the last node.
\end{proof}

A (deterministic, bottom-up) finite tree automaton
over $\Alphabet$-labeled, $\RelCount$-relational ordered \ditree{s}
can be defined as a tuple
$\ClAutomaton = \tuple{\ClStateSet, \tuple{\ClTransFunc_k}_{0\leq k\leq\RelCount}, \ClAcceptSet}$,
where
$\ClStateSet$ is a finite nonempty set of states,
$\ClTransFunc_k \colon \ClStateSet^k \times \Alphabet \to \ClStateSet$
is a transition function of arity~$k$, and
$\ClAcceptSet \subseteq \ClStateSet$ is a set of accepting states.
Such an automaton assigns a state of~$\ClStateSet$
to each node of a given pointed \ditree{},
starting from the leaves and working its way up to the root.
If node~$\Node$ is labeled with symbol~$\Symbol$
and its $k$ children have been assigned the states
$\ClState_1, \dots, \ClState_k$
(following the numbering order of the $k$ first edge relations),
then $\Node$ is assigned the state
$\ClTransFunc_k(\ClState_1, \dots, \ClState_k, \Symbol)$.
Note that leaves are covered by the special case $k = 0$.
Based on this,
the pointed \ditree{} is accepted
if and only if
the state at the root belongs to~$\ClAcceptSet$.
For a more detailed presentation
see, e.g., \cite[\S~3.3]{DBLP:books/ws/automata2012/Loding12}.

\begin{proposition}
  When restricted to the class of pointed ordered \ditree{s},
  forgetful distributed automata are
  strictly more expressive than finite tree automata
  (and thus than monadic second-order logic).
\end{proposition}
\begin{proof}
  To convert a tree automaton
  $\ClAutomaton = \tuple{\ClStateSet, \tuple{\ClTransFunc_k}_{0\leq k\leq\RelCount}, \ClAcceptSet}$
  into a forgetful distributed automaton
  $\Automaton = \tuple{\StateSet,\InitState,\tuple{\TransFunc_\Symbol}_{\Symbol\in\Alphabet},\AcceptSet}$
  that is equivalent to $\ClAutomaton$
  over $\Alphabet$-labeled, $\RelCount$-relational ordered \ditree{s},
  we use a simple generalization of the construction
  in the proof of Proposition~\ref{prp:equivalence-word-automata}:
  $\StateSet = \ClStateSet \cup \set{\Waiting}$,\,
  $\InitState = \Waiting$,\,
  $\AcceptSet = \ClAcceptSet$, and
  \begin{equation*}
    \TransFunc_\Symbol(\vec{\NeighborSet}) =
    \begin{cases*}
      \ClTransFunc_k(\ClState_1, \dots, \ClState_k, \Symbol)
      & if $\vec{\NeighborSet} =
        \bigtuple{\set{\ClState_1}, \dots, \set{\ClState_k}, \emptyset, \dots, \emptyset}$
        for some $\ClState_1,\dots,\ClState_k \in \ClStateSet$, \\
      \Waiting
      & otherwise.
    \end{cases*}
  \end{equation*}

  In contrast,
  a conversion in the other direction is not always possible,
  as can be seen from the following example on binary \ditree{s}.
  Consider the forgetful distributed automaton
  $\Automaton' = \tuple{\set{\Waiting,\Finished,\Accepting}, \Waiting, \TransFunc, \set{\Accepting}}$,
  with
  \begin{equation*}
    \TransFunc(\NeighborSet_1,\NeighborSet_2) =
    \begin{cases*}
      \Waiting
      & if $\NeighborSet_1 = \NeighborSet_2 = \set{\Waiting}$ \\
      \Finished
      & if $\NeighborSet_1, \NeighborSet_2 \in \set{\emptyset, \set{\Finished}}$ \\
      \Accepting
      & otherwise.
    \end{cases*}
  \end{equation*}
  When run on an unlabeled, $2$-relational ordered \ditree{},
  $\Automaton'$ accepts at the root
  precisely if the \ditree{} is \emph{not} perfectly balanced,
  i.e., if there exists a node
  whose left and right subtrees have different heights.
  To achieve this,
  each node starts in the waiting state~$\Waiting$,
  where it remains as long as it has two children
  and those children are also in~$\Waiting$.
  If the \ditree{} is perfectly balanced,
  then all the leaves switch permanently from~$\Waiting$ to~$\Finished$
  in the first round,
  their parents do so in the second round,
  their parents' parents in the third round,
  and so forth,
  until the signal reaches the root.
  Therefore,
  the root will transition directly from~$\Waiting$ to~$\Finished$,
  never visiting state~$\Accepting$,
  and hence the pointed \ditree{} is rejected.
  On the other hand,
  if the \ditree{} is not perfectly balanced,
  then there must be some lowermost internal node~$\Node$
  that does not have two subtrees of the same height
  (in particular, it might have only one child).
  Since its subtrees are perfectly balanced,
  they behave as in the preceding case.
  At some point in time,
  only one of~$\Node$'s children will be in state~$\Waiting$,
  at which point $\Node$ will switch to state~$\Accepting$.
  This triggers an upward-propagating chain reaction,
  eventually causing the root to also visit~$\Accepting$,
  and thus to accept.
  Note that $\Accepting$ is just an intermediate state;
  regardless of whether or not the \ditree{} is perfectly balanced,
  every node will ultimately end up in~$\Finished$.
  
  To prove that $\Automaton'$ is not equivalent to any tree automaton,
  one can simply invoke the pumping lemma for regular tree languages
  to show that the complement language of~$\Automaton'$
  is not recognizable by any tree automaton.
  The claim then follows from the fact
  that regular tree languages are closed under complementation.
\end{proof}
\section{Exploiting forgetfulness}
\label{sec:forgetfulness}
We now give an algorithm deciding
the emptiness problem for forgetful distributed automata
(on arbitrary \digraph{s}).
Its space complexity is linear in the number of states
of the given automaton.
However,
as an uncompressed binary encoding of a distributed automaton
requires space exponential in the number of states,
this results in $\LogSpace$ complexity.
Obviously,
the statement might not hold anymore
if the automaton were instead represented by a more compact device,
such as a logical formula.

\begin{theorem}
  The emptiness problem for forgetful distributed automata is decidable
  in $\LogSpace$.
\end{theorem}
\begin{proof}
  Let
  $\Automaton = \tuple{\StateSet,\InitState,\tuple{\TransFunc_\Symbol}_{\Symbol\in\Alphabet},\AcceptSet}$
  be some forgetful distributed automaton
  over $\Alphabet$-labeled, $\RelCount$-relational \digraph{s}.
  Consider the infinite sequence of sets of states
  $\StateSubset_0, \StateSubset_1, \StateSubset_2 \cdots$
  such that
  $\StateSubset_\Time$ contains precisely those states
  that can be visited by~$\Automaton$
  at some node in some \digraph{} at time~$\Time$.
  That is,
  $\State \in \StateSubset_\Time$
  if and only if
  there exists a pointed \digraph{} $\tuple{\Graph,\Node}$
  such that
  $\Run_\Time(\Node) = \State$,
  where $\Run$ is the run of $\Automaton$ on $\Graph$.
  From this point of view,
  the language of $\Automaton$ is nonempty
  precisely if
  there is some $\Time \in \Natural$
  for which
  $\StateSubset_\Time \cap \AcceptSet \neq \emptyset$.

  By definition,
  we have $\StateSubset_0 = \set{\InitState}$.
  Furthermore,
  exploiting the fact that $\Automaton$ is forgetful,
  we can specify a simple function
  $\BigTransFunc \colon \powerset{\StateSet} \to \powerset{\StateSet}$
  such that
  $\StateSubset_{\Time + 1} = \BigTransFunc(\StateSubset_\Time)$:
  \begin{equation*}
    \BigTransFunc(\StateSubset[1]) =
    \bigsetbuilder
      {\TransFunc_\Symbol(\vec{\StateSubset[2]})}
      {\text{
          $\Symbol \in \Alphabet$ and
          $\vec{\StateSubset[2]} \in (\powerset{\StateSubset[1]})^\RelCount$
        }
      }
  \end{equation*}
  Obviously,
  $\StateSubset_{\Time + 1} \subseteq \BigTransFunc(\StateSubset_\Time)$.
  To see that
  $\StateSubset_{\Time + 1} \supseteq \BigTransFunc(\StateSubset_\Time)$,
  assume we are given a pointed \digraph{}
  $\tuple{\Graph_\State,\Node_\State}$
  for each state~$\State \in \StateSubset_\Time$
  such that
  $\Node_\State$ visits $\State$ at time~$\Time$
  in the run of~$\Automaton$ on~$\Graph_\State$.
  (Such a pointed \digraph{}
  must exist by the definition of $\StateSubset_\Time$.)
  Now,
  for any
  $\Symbol \in \Alphabet$ and
  $\vec{\StateSubset[2]} =
   \tuple{\StateSubset[2]_1, \dots, \StateSubset[2]_\RelCount}
   \in (\powerset{\StateSubset_\Time})^\RelCount$,
  we construct a new \digraph{}~$\Graph$ as follows:
  Starting with a single $\Symbol$-labeled node~$\Node$,
  we add a (disjoint) copy of~$\Graph_\State$
  for each state~$\State$ that occurs in some set~$\StateSubset[2]_k$.
  Then,
  we add a $k$-edge from $\Node_\State$ to $\Node$
  if and only if $\State \in \StateSubset[2]_k$.
  Each node~$\Node_\State$ behaves
  the same way in~$\Graph$ as in~$\Graph_\State$
  because $\Node$ has no influence on its incoming neighbors.
  Since $\Automaton$ is forgetful,
  the state of~$\Node$ at time~$\Time + 1$ depends solely on
  its own label and
  its incoming neighbor's states at time~$\Time$.
  Consequently,
  $\Node$ visits
  the state~$\TransFunc_\Symbol(\vec{\StateSubset[2]})$
  at time $\Time + 1$,
  and thus
  $\TransFunc_\Symbol(\vec{\StateSubset[2]}) \in \StateSubset_{\Time + 1}$.

  Now,
  we know that the sequence
  $\StateSubset_0, \StateSubset_1, \StateSubset_2 \cdots$
  must be eventually periodic
  because its generator function $\BigTransFunc$
  maps the finite set $\powerset{\StateSet}$ to itself.
  Hence,
  it suffices to consider the prefix of length $\card{\powerset{\StateSet}}$
  in order to determine whether
  $\StateSubset_\Time \cap \AcceptSet \neq \emptyset$
  for some $\Time \in \Natural$.
  This leads to the following simple algorithm,
  which decides the emptiness problem for forgetful automata.
  \begin{align*}
    \textsc{Empty}(\Automaton): \quad
    & \StateSubset \gets \set{\InitState} \\
    & \text{\textbf{repeat} at most $\card{\powerset{\StateSet}}$ times}: \\
    & \qquad \StateSubset \gets \BigTransFunc(\StateSubset) \\
    & \qquad \text{\textbf{if} $\StateSubset \cap \AcceptSet \neq \emptyset$}:\;
             \text{\textbf{return} true} \\
    & \text{\textbf{return} false}
  \end{align*}

  It remains to analyze the space complexity of this algorithm.
  For that,
  we assume that the binary encoding of~$\Automaton$
  given to the algorithm
  contains a lookup table
  for each transition function~$\TransFunc_\Symbol$
  and a bit array representing~$\AcceptSet$,
  which amounts to an asymptotic size of
  $\Theta \bigl(
     \card{\Alphabet} \cdot
     \card{\powerset{\StateSet}}^\RelCount \cdot
     \log\card{\StateSet}
   \bigr)$
  input bits.
  To implement the procedure \textsc{Empty},
  we need
  $\card{\StateSet}$ bits of working memory
  to represent the set~$\StateSubset$
  and another $\card{\StateSet}$ bits for the loop counter.
  Furthermore,
  we can compute $\BigTransFunc(\StateSubset[1])$
  for any given set $\StateSubset[1] \subseteq \StateSet$
  by simply iterating over all $\Symbol \in \Alphabet$ and
  $\vec{\StateSubset[2]} \in (\powerset{\StateSet})^\RelCount$,
  and adding $\TransFunc_\Symbol(\vec{\StateSubset[2]})$
  to the returned set
  if all components of~$\vec{\StateSubset[2]}$
  are subsets of~$\StateSubset[1]$.
  This requires
  $\log\card{\Alphabet} + \card{\StateSet} \cdot \RelCount$
  additional bits to keep track of the iteration progress,
  $\Theta \bigl(
     \log\card{\Alphabet} +
     \card{\StateSet} \cdot \RelCount +
     \log\log\card{\StateSet}
   \bigr)$
  bits to store pointers into the lookup tables,
  and $\card{\StateSet}$ bits to store the intermediate result.
  In total,
  the algorithm uses
  $\Theta \bigl(
     \log\card{\Alphabet} + \card{\StateSet} \cdot \RelCount
   \bigr)$
  bits of working memory,
  which is logarithmic in the size of the input.
\end{proof}
\section{Exchanging space and time}
\label{sec:space-time}
In this section,
we first show the undecidability of
the \dipath{}-emptiness problem for arbitrary distributed automata,
and then lift that result to the general emptiness problem.

\begin{theorem}
  \label{thm:dipath-emptiness}
  The \dipath{}-emptiness problem for distributed automata is undecidable.
\end{theorem}
\begin{proof}[Proof sketch]
  We proceed by reduction from the halting problem for Turing machines.
  For our purposes,
  a Turing machine operates deterministically with one head on a single tape,
  which is one-way infinite to the right and initially empty.
  The problem consists of determining
  whether the machine will eventually reach a designated halting state.
  We show a way of encoding the computation of a Turing machine~$\Machine$
  into the run of a distributed automaton~$\Automaton$
  over unlabeled \digraph{s},
  such that the language of $\Automaton$ contains a pointed \dipath{}
  if and only if~$\Machine$ reaches its halting~state.

  Note that since \dipath{s} are oriented,
  the communication between their nodes is only one-way.
  Hence,
  we cannot simply represent (a section of) the Turing tape as a \dipath{}.
  Instead,
  the key idea of our simulation is to exchange the roles of space and time,
  in the sense that
  the space of $\Machine$ is encoded into the time of $\Automaton$,
  and the time of $\Machine$ into the space of $\Automaton$.
  Assuming the language of~$\Automaton$ contains a \dipath{},
  we will think of that \dipath{} as representing the timeline of $\Machine$,
  such that
  each node corresponds to
  a single point in time in the computation of $\Machine$.
  Roughly speaking,
  when running~$\Automaton$,
  the node~$\Node_\Time$ corresponding to time~$\Time$ will “traverse”
  the configuration~$\Config_\Time$ of $\Machine$ at time~$\Time$.
  Here, “traversing” means that
  the sequence of states of $\Automaton$ visited by $\Node_\Time$
  is an encoding of $\Config_\Time$ read from left to right,
  supplemented with some additional bookkeeping information.

  The first element of the \dipath{}, node~$\Node_0$,
  starts by visiting a state of $\Automaton$ representing
  an empty cell that is currently read by $\Machine$ in its initial state.
  Then it transitions to another state that simply represents an empty cell,
  and remains in such a state forever after.
  Thus $\Node_0$ does indeed “traverse” $\Config_0$.
  We will show that
  it is also possible for any other node~$\Node_\Time$
  to “traverse” its corresponding configuration~$\Config_\Time$,
  based on the information it receives from $\Node_{\Time-1}$.
  In order for this to work,
  we shall give $\Node_{\Time-1}$ a head start of two cells,
  so that $\Node_\Time$ can compute the content of cell $\Cell$ in $\Config_\Time$
  based on the contents of cells $\Cell-1$,\, $\Cell$ and $\Cell+1$ in $\Config_{\Time-1}$.

  Node~$\Node_\Time$ enters an accepting state of $\Automaton$ precisely if
  it “sees” the halting state of $\Machine$ during its “traversal” of $\Config_\Time$.
  Hence,
  $\Automaton$ accepts the pointed \dipath{} of length~$\Time$
  if and only if $\Machine$ reaches its halting state at time~$\Time$.

  We now describe the inner workings of $\Automaton$ in a semi-formal way.
  In parallel,
  the reader might want to have a look at Figure~\ref{fig:exchange-space-time},
  which illustrates the construction by means of an example.
  Let $\Machine$ be represented by the tuple
  $\tuple{\MStateSet,\MSymbolSet,\MInitState,\MBlank,\MTransFunc,\MHaltState}$,
  where
  $\MStateSet$ is the set of states,
  $\MSymbolSet$ is the tape alphabet,
  $\MInitState$ is the initial state,
  $\MBlank$ is the blank symbol,
  $\MTransFunc\colon
   (\MStateSet\setminus\set{\MHaltState}) \times \MSymbolSet \to
   \MStateSet \times \MSymbolSet \times \set{\Left,\Right}$
  is the transition function, and
  $\MHaltState$ is the halting state.
  From this,
  we construct $\Automaton$ as $\tuple{\StateSet,\InitState,\TransFunc,\AcceptSet}$,
  with
  the state set
  $\StateSet = (\set{\Waiting} \,\cup\, (\MStateSet\times\MSymbolSet) \,\cup\, \MSymbolSet)^3$,
  the initial state
  $\InitState = \tuple{\Waiting,\Waiting,\Waiting}$,
  the transition function $\TransFunc$ specified informally below, and
  the accepting set $\AcceptSet$ that contains precisely those states
  that have $\MHaltState$ in their third component.
  In keeping with the intuition that each node of the \dipath{}
  “traverses” a configuration of $\Machine$,
  the third component of its state indicates
  the content of the “currently visited” cell $\Cell$.
  The two preceding components keep track of the recent history,
  i.e.,
  the second component always holds the content of the previous cell $\Cell-1$,
  and the first component that of $\Cell-2$.
  In the following explanation,
  we concentrate on updating the third component,
  tacitly assuming that the other two are kept up to date.
  The special symbol $\Waiting$ indicates that no cell has been “visited”,
  and we say that a node is in the waiting phase
  while its third component is~$\Waiting$.

  In the first round,
  $\Node_0$ sees that it does not have any incoming neighbor,
  and thus exits the waiting phase
  by setting its third component to $\tuple{\MInitState,\MBlank}$,
  and after that, it sets it to $\MBlank$ for the remainder of the run.
  Every other node $\Node_\Time$ remains in the waiting phase
  as long as its incoming neighbor's second component is $\Waiting$.
  This ensures a delay of two cells with respect to $\Node_{\Time-1}$.
  Once $\Node_\Time$ becomes active,
  given the \emph{current} state $\tuple{c_1,c_2,c_3}$ of $\Node_{\Time-1}$,
  it computes the third component $d_3$ of its own \emph{next} state $\tuple{d_1,d_2,d_3}$
  as follows:
  If none of the components $c_1$, $c_2$, $c_3$ “contain the head of $\Machine$”,
  i.e., if none of them lie in $\MStateSet\times\MSymbolSet$,
  then it simply sets~$d_3$ to be equal to~$c_2$.
  Otherwise,
  a computation step of $\Machine$ is simulated in the natural way.
  For instance,
  if $c_3$ is of the form $\tuple{\MState,\MSymbol}$,
  and $\MTransFunc(\MState,\MSymbol) = \tuple{\MState',\MSymbol',\Left}$,
  then $d_3$ is set to $\tuple{\MState',c_2}$.
  This corresponds to the case where, at time $\Time-1$,
  the head of $\Machine$ is located to the right of $\Node_\Time$'s next “position”
  and moves to the left.
  As another example,
  if $c_2$ is of the form $\tuple{\MState,\MSymbol}$,
  and $\MTransFunc(\MState,\MSymbol) = \tuple{\MState',\MSymbol',\Right}$,
  then $d_3$ is set to $\MSymbol'$.
  The remaining cases are handled analogously.

  Note that, thanks to the two-cell delay between adjacent nodes,
  the head of $\Machine$ always “moves forward” in the time of $\Automaton$,
  although it may move in both directions with respect to the space of $\Machine$
  (see Figure~\ref{fig:exchange-space-time}).
\end{proof}

\begin{figure}[tbp]
  \centering
  \begin{tikzpicture}[
    thick,
    every matrix/.style = {matrix of math nodes, nodes in empty cells, row sep={4.55ex,between origins}, inner sep=0ex,
                           nodes={draw,minimum size=3ex,inner sep=0ex,anchor=center}},
    b/.style  = {fill=black},               
    g/.style  = {fill=lightgray},           
    hg/.style = {draw=darkgray},            
    gg/.style = {g, hg, line width=0.33ex}, 
    G/.style  = {g, hg, line width=0.41ex}, 
    ww/.style = {hg, line width=0.33ex},    
    W/.style  = {hg, line width=0.41ex},    
    plain/.style = {draw opacity=0}
  ]
  \newcommand{\x}[1]{\mathbf{#1}}
  \matrix[column sep={3ex,between origins}] (machine) {
    \x0       &           &            &        &        \\
    |[g]|     & \x1       &            &        &        \\
    |[gg]|    & |[gg]|    & |[ww]| \x2 & |[ww]| & |[ww]| \\
    |[g]|     & |[g]| \x1 & |[g]|      &        &        \\
    |[g]| \x0 & |[g]|     & |[g]|      &        &        \\
              & |[g]| \x3 & |[g]|      &        &        \\
  };
  \matrix[right=0.4ex of machine] (machine-dots) {
    |[plain]|\cdots \\ |[plain]|\cdots \\ |[plain]|\cdots \\ |[plain]|\cdots \\ |[plain]|\cdots \\ |[plain]|\cdots \\
  };
  \matrix[right=7.5ex of machine-dots, column sep={4.45ex,between origins}, nodes={circle}] (automaton) {
    |[b]| & \x0   &       &       &       &       &       &           &           &           &       &       &          \\
    |[b]| & |[b]| & |[b]| & |[g]| & \x1   &       &       &           &           &           &       &       &          \\
    |[b]| & |[b]| & |[b]| & |[b]| & |[b]| & |[G]| & |[G]| & |[W]| \x2 & |[W]|     & |[W]|     &       &       &          \\
    |[b]| & |[b]| & |[b]| & |[b]| & |[b]| & |[b]| & |[b]| & |[g]|     & |[g]| \x1 & |[g]|     &       &       &          \\
    |[b]| & |[b]| & |[b]| & |[b]| & |[b]| & |[b]| & |[b]| & |[b]|     & |[b]|     & |[g]| \x0 & |[g]| & |[g]| &          \\
    |[b]| & |[b]| & |[b]| & |[b]| & |[b]| & |[b]| & |[b]| & |[b]|     & |[b]|     & |[b]|     & |[b]| &       & |[g]|\x3 \\
  };
  \matrix[right=1ex of automaton] (automaton-dots) {
    |[plain]|\cdots \\ |[plain]|\cdots \\ |[plain]|\cdots \\ |[plain]|\cdots \\ |[plain]|\cdots \\ |[plain]|\cdots \\
  };
  \foreach \x [evaluate = \x as \xplusone using int(\x+1)] in {1,...,5}
    \foreach \y in {1,...,13}
      \draw[->] (automaton-\x-\y) -- (automaton-\xplusone-\y);
  \node[above=2.5ex of machine.north west,anchor=mid west] (machine-space) {\small{space}};
  \draw[->] ([xshift=0.5ex]machine-space.mid east) -- ++(5ex,0);
  \node[left=2.5ex of machine.north west,anchor=mid west,rotate=-90] (machine-time) {\small{time}};
  \draw[->] ([yshift=-0.5ex]machine-time.mid east) -- ++(0,-6ex);
  \node[above=2.5ex of automaton.north west,anchor=mid west] (automaton-time) {\small{time}};
  \draw[->] ([xshift=0.5ex]automaton-time.mid east) -- ++(5ex,0);
  \node[left=2.5ex of automaton.north west,anchor=mid west,rotate=-90] (automaton-space) {\small{space}};
  \draw[->] ([yshift=-0.5ex]automaton-space.mid east) -- ++(0,-5ex);
  \node[below=3ex of machine,anchor=mid] {\small{\emph{Turing machine}}};
  \node[below=3ex of automaton,anchor=mid] {\small{\emph{Distributed automaton}}};
\end{tikzpicture}
  \caption{
    Exchanging space and time to prove Theorem~\ref{thm:dipath-emptiness}.
    The left-hand side depicts the computation of a
    Turing machine with state set $\set{\mathbf{0},\mathbf{1},\mathbf{2},\mathbf{3}}$
    and tape alphabet
    $\set{
      \tikz{\node[minimum width=2ex]{};\node[minimum width=1ex,draw]{};},
      \tikz{\node[minimum width=2ex]{};\node[minimum width=1ex,draw,fill=lightgray]{};}
    }$.
    On the right-hand side,
    this machine is simulated by a distributed automaton run on a \dipath{}.
    Waiting nodes are represented in black,
    whereas active nodes display the content of the “currently visited” cell of the Turing machine
    (i.e., only the third component of the states is shown).
  }
  \label{fig:exchange-space-time}
\end{figure}

To infer from Theorem~\ref{thm:dipath-emptiness}
that the general emptiness problem for distributed automata
is also undecidable,
we now introduce the notion of
\defd{monovisioned} automata,
which have the property
that nodes “expect” to see
no more than one state in their incoming neighborhood at any given time.
More precisely,
a distributed automaton
$\Automaton = \tuple{\StateSet,\InitFunc,\TransFunc,\AcceptSet}$
is monovisioned if it has a rejecting sink state
$\RejectState \in \StateSet \setminus \AcceptSet$,
such that
$\TransFunc(\State,\NeighborSet) = \RejectState$
whenever $\card{\NeighborSet} > 1$
or $\RejectState \in \NeighborSet$
or $\State = \RejectState$,
for all $\State \in \StateSet$ and $\NeighborSet \subseteq \StateSet$.
Obviously,
for every distributed automaton,
we can construct a monovisioned automaton
that has the same acceptance behavior on \dipath{s}.
Furthermore,
as shown by means of the next two lemmas,
the emptiness problem for monovisioned automata
is equivalent to its restriction to \dipath{s}.
All put together,
we get the desired reduction
from the \dipath{}-emptiness problem to the general emptiness problem.

\begin{lemma}
  \label{lem:tree-model-property}
  The language of a distributed automaton is nonempty
  if and only if it contains a pointed \ditree{}.
\end{lemma}
\begin{proof}[Proof sketch]
  We slightly adapt the notion of \emph{tree-unraveling},
  which is a standard tool in modal logic
  (see, e.g., \cite[Def.~4.51]{BlackburnRV02} or \cite[\S~3.2]{BlackburnB07}).
  Consider any distributed automaton~$\Automaton$.
  Assume that $\Automaton$ accepts some pointed \digraph{}
  $\tuple{\Graph,\Node[2]}$,
  and let~$\Time \in \Natural$ be the first point in time
  at which $\Node$ visits an accepting state.
  Based on that,
  we can easily construct a pointed \ditree{}
  $\tuple{\Graph',\Node[2]'}$
  that is also accepted by~$\Automaton$.
  First of all,
  the root~$\Node[2]'$ of~$\Graph'$
  is chosen to be a copy of~$\Node[2]$.
  On the next level of the \ditree{},
  the incoming neighbors of~$\Node[2]'$
  are chosen to be fresh copies
  $\Node[1]'_1,\dots,\Node[1]'_n$
  of $\Node[2]$'s incoming neighbors
  $\Node[1]_1,\dots,\Node[1]_n$.
  Similarly,
  the incoming neighbors of
  $\Node[1]'_1,\dots,\Node[1]'_n$
  are fresh copies of the incoming neighbors of
  $\Node[1]_1,\dots,\Node[1]_n$.
  If $\Node[1]_i$ and $\Node[1]_j$
  have incoming neighbors in common,
  we create distinct copies of those neighbors
  for $\Node[1]_i'$ and $\Node[1]_j'$.
  This process is iterated until we obtain
  a \ditree{} of height~$\Time$.
  It is easy to check
  that~$\Node$ and~$\Node'$ visit
  the same sequence of states
  $\State_0, \State_1, \dots, \State_\Time$
  during the first~$\Time$ communication rounds.
\end{proof}

\begin{lemma}
  \label{lem:dipath-model-property}
  The language of a monovisioned distributed automaton is nonempty
  if and only if it contains a pointed \dipath{}.
\end{lemma}
\begin{proof}[Proof sketch]
  Consider any monovisioned distributed automaton~$\Automaton$
  whose language is nonempty.
  By Lemma~\ref{lem:tree-model-property},
  $\Automaton$ accepts some pointed \ditree{} $\tuple{\Graph,\Node[2]}$.
  Let~$\Time \in \Natural$ be the first point in time
  at which $\Node[2]$ visits an accepting state.
  Now,
  it is easy to prove by induction that
  for all $i \in \set{0, \dots, \Time}$,
  sibling nodes at depth~$i$
  traverse the same sequence of states
  $\State_0, \State_1, \dots, \State_{\Time - i}$
  between times~$0$ and~$\Time - i$,
  and this sequence does not contain the rejecting state~$\RejectState$.
  Thus,
  $\Automaton$ also accepts
  any \dipath{} from some node at depth~$\Time$ to the root.
\end{proof}
\section{Timing a firework show}
\label{sec:fireworks}
We now show
that the emptiness problem is undecidable
even for quasi-acyclic automata.
This also provides an alternative, but more involved
undecidability proof for the general case.

A distributed automaton
$\Automaton = \tuple{\StateSet,\InitFunc,\TransFunc,\AcceptSet}$
is said to be \defd{quasi-acyclic}
if its state diagram does not contain any directed cycles,
except for self-loops.
More formally,
this means that for every sequence
$\State_1, \State_2, \dots, \State_n$ of states in~$\StateSet$
such that
$\State_1 = \State_n$ and
$\TransFunc(\State_i,\NeighborSet_i) = \State_{i+1}$
for some $\NeighborSet_i \subseteq \StateSet$,
it must hold that all states of the sequence are the same.
Notice that our proof of Theorem~\ref{thm:dipath-emptiness}
does not go through
if we consider only quasi-acyclic automata.

It is straightforward to see that
quasi-acyclicity is preserved under a standard product construction,
similar to the one employed for finite automata on words.
Hence, we have the following closure property,
which will be used in the subsequent undecidability proof.

\begin{lemma}
  \label{lem:closure-quasi-acyclic}
  The class of languages recognizable by quasi-acyclic distributed automata
  is closed under union and intersection.
\end{lemma}

\begin{theorem}
  \label{thm:emptiness-quasi-acyclic}
  The emptiness problem for quasi-acyclic distributed automata is undecidable.
\end{theorem}
\begin{proof}[Proof sketch]
  We show this by reduction from Post's correspondence problem (PCP).
  An instance $\Instance$ of PCP consists of
  a collection of pairs of nonempty finite words
  $\tuple{\UpperWord_\Index,\LowerWord_\Index}_{\Index\in\IndexSet}$
  over the alphabet $\set{0,1}$,
  indexed by some finite set of integers $\IndexSet$.
  It is convenient to view each pair
  $\tuple{\UpperWord_\Index,\LowerWord_\Index}$
  as a domino tile
  labeled with $\UpperWord_\Index$ on the upper half
  and $\LowerWord_\Index$ on the lower half.
  The problem is to decide if there exists a nonempty sequence
  $\Solution = \tuple{\Index_1, \dots, \Index_\SolutionSize}$
  of indices in~$\IndexSet$,
  such that the concatenations
  $\UpperWord_\Solution = \UpperWord_{\Index_1} \!\cdots \UpperWord_{\Index_\SolutionSize}$ and 
  $\LowerWord_\Solution = \LowerWord_{\Index_1} \!\cdots \LowerWord_{\Index_\SolutionSize}$
  are equal.
  We construct a quasi-acyclic automaton $\Automaton$
  whose language is nonempty if and only if
  $\Instance$ has such a solution~$\Solution$.

  Metaphorically speaking,
  our construction can be thought of as a perfectly timed “firework show”,
  whose only “spectator” will see a putative solution
  $\Solution = \tuple{\Index_1, \dots, \Index_\SolutionSize}$,
  and be able to check whether it is indeed a valid solution of $\Instance$.
  Our “spectator” is the distinguished node~$\Root$ of the pointed \digraph{}
  on which $\Automaton$ is run.
  We assume that $\Root$ has $\SolutionSize$ incoming neighbors,
  one for each element of $\Solution$.
  Let $\Node_\IIndex$ denote the neighbor corresponding to $\Index_\IIndex$,
  for $1 \leq \IIndex \leq \SolutionSize$.
  Similarly to our proof of Theorem~\ref{thm:dipath-emptiness},
  we use the time of $\Automaton$ to represent
  the spatial dimension of the words $\UpperWord_\Solution$ and $\LowerWord_\Solution$.
  On an intuitive level,
  $\Root$ will “witness” simultaneous left-to-right traversals of
  $\UpperWord_\Solution$ and $\LowerWord_\Solution$,
  advancing by one bit per time step,
  and it will check that the two words match.
  It is the task of each node $\Node_\IIndex$
  to send to $\Root$ the required bits of
  the subwords $\UpperWord_{\Index_\IIndex}$ and $\LowerWord_{\Index_\IIndex}$
  at the appropriate times.
  In keeping with the metaphor of fireworks,
  the correct timing can be achieved by attaching to $\Node_\IIndex$
  a carefully chosen “fuse”,
  which is “lit” at time $0$.
  Two separate “fire” signals will travel at different speeds
  along this (admittedly sophisticated) “fuse”,
  and once they reach $\Node_\IIndex$,
  they trigger the “firing” of
  $\UpperWord_{\Index_\IIndex}$ and $\LowerWord_{\Index_\IIndex}$, respectively.

  We now go into more details.
  Using the labeling of the input graph,
  the automaton $\Automaton$ distinguishes between
  $2\InstanceSize+1$ different types of nodes:
  two types $\Index$ and $\Index'$ for each index $\Index\in\IndexSet$,
  and one additional type $\Spectator$ to identify the “spectator”.
  Motivated by Lemma~\ref{lem:tree-model-property},
  we suppose that the input graph is a pointed \ditree{},
  with a very specific shape that encodes a putative solution
  $\Solution = \tuple{\Index_1, \dots, \Index_\SolutionSize}$.
  An example illustrating the following description
  of such a \ditree{}-encoding
  is given in Figure~\ref{fig:firework-show}.
  Although $\Automaton$ is not able to enforce
  all aspects of this particular shape,
  we will make sure that it accepts such a structure
  if its language is nonempty.
  The root (and distinguished node) $\Root$
  is the only node of type $\Spectator$.
  Its children $\Node_1, \dots, \Node_\SolutionSize$
  are of types $\Index_1, \dots, \Index_\SolutionSize$,
  respectively.
  The “fuse” attached to each child~$\Node_\IIndex$
  is a chain of $\IIndex-1$ nodes
  that represents the multiset of indices occurring in
  the $(\IIndex-1)$-prefix of $\Solution$.
  More precisely,
  there is an induced \dipath{}
  $\Node_{\IIndex,1} \rightarrow \cdots\, \Node_{\IIndex,\IIndex-1} \rightarrow \Node_\IIndex$,
  such that the multiset of types of the nodes
  $\Node_{\IIndex,1}, \dots, \Node_{\IIndex,\IIndex-1}$
  is equal to the multiset of indices occurring in
  $\tuple{\Index_1, \dots, \Index_{\IIndex-1}}$.
  We do not impose any particular order on those nodes.
  Finally,
  each node of type $\Index \in \IndexSet$
  also has an incoming chain of nodes of type $\Index'$
  \mbox{(depicted in gray in Figure~\ref{fig:firework-show})},
  whose length corresponds exactly to
  the product of the types occurring on the part of the “fuse”
  below that node.
  That is,
  if we define the alias $\Node_{\IIndex,\IIndex} \defeq \Node_\IIndex$,
  then for every node $\Node_{\IIndex,\IIIndex}$ of type $\Index \in \IndexSet$,
  there is an induced \dipath{}
  $\Node_{\IIndex,\IIIndex,1} \rightarrow \cdots\,
   \Node_{\IIndex,\IIIndex,\PrimeProduct} \rightarrow
   \Node_{\IIndex,\IIIndex}$,
  where all the nodes
  $\Node_{\IIndex,\IIIndex,1},\dots,\Node_{\IIndex,\IIIndex,\PrimeProduct}$
  are of type $\Index'$,
  and the number $\PrimeProduct$ is equal to the product of the types
  of the nodes $\Node_{\IIndex,1},\dots,\Node_{\IIndex,\IIIndex-1}$
  (which is $1$ if $\IIIndex = 1$).
  We shall refer to such a chain
  $\Node_{\IIndex,\IIIndex,1},\dots,\Node_{\IIndex,\IIIndex,\PrimeProduct}$
  as a “side fuse”.

  \begin{figure}[tbp]
    \centering
    \begin{tikzpicture}[
    thick,
    edge from parent/.style = {draw, <-},
    level 1/.style = {level distance=9ex},
    level 2/.style = {level distance=8ex},
    sibling distance = 22ex,
    every node/.style = {inner sep=0ex},
    nodraw/.style = {draw=none, fill=none, text=black},
    bold/.style = {draw=darkgray, line width=0.41ex},
    main node/.style = {draw, circle, inner sep=0ex, minimum size=3ex},
    root node/.style = {main node, fill=black, text=white},
    side node/.style = {main node, draw=gray, fill=gray, text=white},
    side line/.style = {grow=right, every node/.style={side node}, every child/.style={}, level distance=5ex},
    brace/.style = {decorate, decoration={brace,mirror,raise=1ex}, darkgray},
    brace label/.style = {minimum size=0ex, below=2.2ex},
    domino/.style = {draw, rectangle split, rectangle split parts=2, rounded corners=0.6ex, inner sep=1ex, minimum width=5.5ex}
  ]
  \newcommand{\x}[1]{\mathbf{#1}}
  \path[every node/.style=main node] node[root node] (root) {$\x\epsilon$}
     child[xshift=6ex] {node[bold] {$\x5$} [child anchor=50,parent anchor=west,
                                          every child/.style={grow=down,child anchor=border,parent anchor=border}]
       child[side line,parent anchor=border] {node {$5'$}}
     }
     child {node[bold] {$\x3$} [every child/.style={grow=down}]
       child {node {$5$}
         child[side line] {node {$5'$}}
       }
       child[side line] {node {$3'$} child {node[nodraw] {\footnotesize\,\dots} child {node {$3'$}}}}
     }
     child {node[bold] {$\x7$} [every child/.style={grow=down}]
       child {node {$5$}
         child {node {$3$}
           child[side line] {node {$3'$}}
         }
         child[side line] {node {$5'$} child {node {$5'$} child {node {$5'$}}}}
       }
       child[side line] {node {$7'$} child {node[nodraw] {\footnotesize\,\dots} child {node {$7'$}}}}
     }
     child {node[bold] {$\x3$} [child anchor=north west,parent anchor=east,
                              every child/.style={grow=down,child anchor=border,parent anchor=border}]
       child {node {$3$}
         child {node {$7$}
           child {node {$5$}
             child[side line] {node {$5'$}}
           }
           child[side line] {node {$7'$} child {node[nodraw] {\footnotesize\,\dots} child {node {$7'$}}}}
         }
         child[side line] {node {$3'$} child {node[nodraw] {\footnotesize\,\dots} child {node {$3'$}}}}
       }
       child[side line] {node {$3'$} child {node[nodraw] {\footnotesize\,\dots} child {node {$3'$}}}}
     }
  ;
  \draw[brace] (root-2-2.south west) -- node[brace label] {$5$} (root-2-2-1-1.south east);
  \draw[brace] (root-3-2.south west) -- node[brace label] {$3 \times 5$} (root-3-2-1-1.south east);
  \draw[brace] (root-3-1-2.south west) -- node[brace label] {$3$} (root-3-1-2-1-1.south east);
  \draw[brace] (root-4-2.south west) -- node[brace label] {$5 \times 7 \times 3$} (root-4-2-1-1.south east);
  \draw[brace] (root-4-1-2.south west) -- node[brace label] {$5 \times 7$} (root-4-1-2-1-1.south east);
  \draw[brace] (root-4-1-1-2.south west) -- node[brace label] {$5$} (root-4-1-1-2-1-1.south east);
  \matrix[xshift=-20ex,yshift=-28ex,row sep=1ex,column sep=1ex,matrix of nodes] {
    \node[domino] {$010$ \nodepart{two} $0$};  &
    \node[domino] {$00$ \nodepart{two} $100$}; & 
    \node[domino] {$11$ \nodepart{two} $01$};  &
    \node[domino] {$00$ \nodepart{two} $100$}; \\
    $\mathbf{5}$ & $\mathbf{3}$ & $\mathbf{7}$ & $\mathbf{3}$ \\
  };
\end{tikzpicture}
    \caption{
      Timing a “firework show” to prove Theorem~\ref{thm:emptiness-quasi-acyclic}.
      The domino tiles on the bottom-left visualize the solution $\tuple{5,3,7,3}$
      for the instance
      $\set{3 {\,\mapsto\,} \tuple{00,100},\,
        5 {\,\mapsto\,} \tuple{010,0},\,
        7 {\,\mapsto\,} \tuple{11,01}}$
      of PCP.
      This solution is encoded into the labeled \ditree{} above,
      with node types $\Spectator$, $3$, $5$, $7$, $3'$, $5'$, $7'$.
      Each domino is represented by a bold-highlighted white node of the appropriate type.
      The “fuse” of such a bold node consists of
      the chain of white nodes below it,
      which lists the indices of the preceding dominos in an arbitrary order.
      Each white node also has a gray “side fuse”
      whose length is equal to the product of the white types occurring below that node.
      The “firework show” observed at the root will feature two simultaneous bitstreams,
      which both represent the sequence $010001100$.
    }
    \label{fig:firework-show}
  \end{figure}

  The automaton $\Automaton$ has to perform two tasks simultaneously:
  First,
  assuming it is run on a \ditree{}-encoding of a sequence $\Solution$,
  exactly as specified above,
  it must verify that $\Solution$ is a valid solution,
  i.e., that the words $\UpperWord_\Solution$ and $\LowerWord_\Solution$ match.
  Second,
  it must ensure that the input graph is indeed
  sufficiently similar to such a \ditree{}-encoding.
  In particular,
  it has to check that
  the “fuses” used for the first task are consistent with each other.
  Since, by Lemma~\ref{lem:closure-quasi-acyclic},
  quasi-acyclic distributed automata are closed under intersection,
  we can consider the two tasks separately,
  and implement them using two independent automata $\Automaton_1$ and~$\Automaton_2$.
  In the following,
  we describe both devices in a rather informal manner.
  The important aspect to note is
  that they can be easily formalized using quasi-acyclic state diagrams.

  We start with $\Automaton_1$,
  which verifies the solution $\Solution$.
  It takes into account
  only nodes with types in $\IndexSet \cup \set{\Spectator}$
  (thus ignoring the gray nodes in Figure~\ref{fig:firework-show}).
  At nodes of type $\Index \in \IndexSet$,
  the states of $\Automaton_1$ have two components,
  associated with the upper and lower halves of the
  domino~$\tuple{\UpperWord_\Index,\LowerWord_\Index}$.
  If a node of type $\Index$ sees
  that it does not have any incoming neighbor,
  then the upper and lower components of its state
  immediately start traversing sequences of substates
  representing the bits of $\UpperWord_\Index$ and $\LowerWord_\Index$,
  respectively.
  Since those substates must keep track of
  the respective positions within $\UpperWord_\Index$ and $\LowerWord_\Index$,
  none of them can be visited twice.
  After that,
  both components loop forever on a special substate $\Finished$,
  which indicates the end of transmission.
  The other nodes of type $\Index$
  keep each of their two components in a waiting status,
  indicated by another substate~$\Waiting$,
  until the corresponding component of their incoming neighbor
  reaches its last substate before~$\Finished$.
  This constitutes the aforementioned “fire” signal.
  Thereupon,
  they start traversing the same sequences of substates as in the previous case.
  Note that both components are updated independently of each other,
  hence there can be an arbitrary time lag between
  the “traversals” of $\UpperWord_\Index$ and $\LowerWord_\Index$.
  Now, assuming the “fuse” of each node $\Node_\IIndex$
  really encodes the multiset of indices occurring in
  $\tuple{\Index_1, \dots, \Index_{\IIndex-1}}$,
  the delay accumulated along that “fuse” will be such that
  $\Node_\IIndex$ starts “traversing”
  $\UpperWord_{\Index_\IIndex}$ and $\LowerWord_{\Index_\IIndex}$
  at the points in time corresponding to their respective starting positions
  within $\UpperWord_\Solution$ and $\LowerWord_\Solution$.
  That is,
  for~$\UpperWord_{\Index_\IIndex}$ it starts at time
  $\length{\UpperWord_{\Index_1} \!\cdots \UpperWord_{\Index_{\IIndex-1}}} + 1$,
  and for~$\LowerWord_{\Index_\IIndex}$ at time
  $\length{\LowerWord_{\Index_1} \!\cdots \LowerWord_{\Index_{\IIndex-1}}} + 1$.
  Consequently,
  in each round
  $\Time \leq \min\set{\card{\UpperWord_\Solution}, \card{\LowerWord_\Solution}}$,
  the root $\Root$ receives the $\Time$-th bits of
  $\UpperWord_\Solution$ and $\LowerWord_\Solution$.
  At most two distinct children send bits at the same time,
  while the others remain
  in some state $\State \in \set{\Waiting,\Finished}^2$.
  With this,
  the behavior of $\Automaton_1$ at $\Root$ is straightforward:
  It enters its only accepting state precisely if
  all of its children
  have reached the state~$\tuple{\Finished,\Finished}$
  and it has never seen any mismatch between the upper and lower bits.

  We now turn to $\Automaton_2$,
  whose job is to verify that
  the “fuses” used by $\Automaton_1$ are reliable.
  Just like $\Automaton_1$,
  it works under the assumption
  that the input graph is a \ditree{} as specified previously,
  but with significantly reduced guarantees:
  The root could now have an arbitrary number of children,
  the “fuses” and “side fuses” could be of arbitrary lengths,
  and each “fuse” could represent an arbitrary multiset of indices in $\IndexSet$.
  Again using an approach reminiscent of fireworks,
  we devise a protocol in which
  each child~$\Node$ will send
  two distinct signals to the root~$\Root$.
  The first signal $\Signal_1$ indicates that
  the current time~$\Time$ is equal to the product of
  the types of all the nodes on $\Node$'s “fuse”.
  Similarly,
  the second signal $\Signal_2$ indicates that
  the current time is equal to that same product
  multiplied by $\Node$'s own type.
  To achieve this,
  we make use of the “side fuses”,
  along which two additional signals
  $\SideSignal_1$ and $\SideSignal_2$ are propagated.
  For each node of type $\Index \in \IndexSet$,
  the nodes of type $\Index'$ on the corresponding “side fuse”
  operate in a way such that
  $\SideSignal_1$ advances by one node per time step,
  whereas $\SideSignal_2$ is delayed by $\Index$ time units at every node.
  Hence,
  $\SideSignal_1$ travels $\Index$ times faster than~$\SideSignal_2$.
  Building on that,
  each node $\Node$ of type~$\Index$
  (not necessarily a child of the root)
  sends $\Signal_1$ to its parent,
  either at time~$1$, if it does not have any predecessor on the “fuse”,
  or one time unit before receiving $\Signal_2$ from its predecessor.
  The latter is possible,
  because the predecessor also sends a pre-signal~$\PreSignal_2$
  before sending $\Signal_2$.
  Then,
  $\Node$ checks that signal $\SideSignal_1$ from its “side fuse”
  arrives exactly at the same time as
  $\Signal_2$ from its predecessor,
  or at time $1$ if there is no predecessor.
  Otherwise,
  it immediately enters a rejecting state.
  This will guarantee, by induction, that the length of the “side fuse”
  is equal to the product of the types on the “fuse” below.
  Finally,
  two rounds prior to receiving $\SideSignal_2$,
  while that signal is still being delayed
  by the last node on the “side fuse”,
  $\Node$ first sends the pre-signal $\PreSignal_2$\!,
  and then the signal $\Signal_2$ in the following round.
  For this to work,
  we assume that
  each node on the “side fuse” waits for at least two rounds
  between receiving $\SideSignal_2$ from its predecessor
  and forwarding the signal to its successor,
  i.e., all indices in $\IndexSet$ must be strictly greater than~$2$.
  Due to the delay accumulated by $\SideSignal_2$ along the “side fuse”,
  the time at which $\Signal_2$ is sent corresponds precisely to
  the length of the “side fuse” multiplied by $\Index$.

  Without loss of generality,
  we require that
  the set of indices $\IndexSet$ contains only prime numbers
  (as in Figure~\ref{fig:firework-show}).
  Hence,
  by the unique-prime-factorization theorem,
  each multiset of numbers in $\IndexSet$ is uniquely determined
  by the product of its elements.
  This leads to a simple verification procedure
  performed by $\Automaton_2$ at the root:
  At time $1$,
  node $\Root$ checks that it receives $\Signal_1$ and not $\Signal_2$.
  After that,
  it expects to never again see $\Signal_1$ without $\Signal_2$,
  and remains in a loop as long as it gets
  either no signal at all or both $\Signal_1$ and $\Signal_2$.
  Upon receiving $\Signal_2$ alone,
  it exits the loop
  and verifies that all of its children have sent both signals,
  which is apparent from the state of each child.
  The root rejects immediately
  if any of the expectations above are violated,
  or if two nodes with different types
  send the same signal at the same time.
  Otherwise,
  it enters an accepting state after leaving the loop.
  Now,
  consider the sequence
  $\TimeSequence = \tuple{\Time_1,\dots,\Time_{\SolutionSize+1}}$
  of rounds in which $\Root$ receives
  at least one of the signals $\Signal_1$ and $\Signal_2$.
  It is easy to see by induction on $\TimeSequence$ that
  successful completion of the procedure above
  ensures that there is a sequence
  $\Solution = \tuple{\Index_1, \dots, \Index_\SolutionSize}$
  of indices in $\IndexSet$
  with the following properties:
  For each $\IIndex \in \set{1,\dots,\SolutionSize}$,
  the root has at least one child $\Node_\IIndex$ of type $\Index_\IIndex$
  that sends $\Signal_1$ at time~$\Time_\IIndex$
  and $\Signal_2$ at time~$\Time_{\IIndex+1}$,
  and the “fuse” of $\Node_\IIndex$ encodes precisely
  the multiset of indices occurring in
  $\tuple{\Index_1, \dots, \Index_{\IIndex-1}}$.
  Conversely,
  each child of $\Root$ can be associated in the same manner
  with a unique element of $\Solution$.

  To conclude our proof,
  we have to argue that the automaton $\Automaton$,
  which simulates $\Automaton_1$ and $\Automaton_2$ in parallel,
  accepts some labeled pointed \digraph{}
  if and only if
  $\Instance$ has a solution~$\Solution$.
  The “if” part is immediate,
  since, by construction,
  $\Automaton$ accepting a \ditree{}-encoding of $\Solution$
  is equivalent to $\Solution$ being a valid solution of $\Instance$.
  To show the “only if” part,
  we start with a pointed \digraph{} accepted by $\Automaton$,
  and incrementally transform it into a \ditree{}-encoding of a solution $\Solution$,
  while maintaining acceptance by $\Automaton$:
  First of all,
  by Lemma~\ref{lem:tree-model-property},
  we may suppose that the \digraph{} is a \ditree{}.
  Its root must be of type~$\Spectator$,
  since $\Automaton$ would not accept otherwise.
  Next,
  we require that $\Automaton$ raises an alarm
  at nodes that see an unexpected set of states in their incoming neighborhood,
  and that this alarm is propagated up to the root,
  which then reacts by entering a rejecting sink state.
  This ensures that the repartition of types is consistent with our specification;
  for example,
  that the children of a node of type~$\Index'$ must be of type~$\Index'$ themselves.
  We now prune the \ditree{}
  in such a way that
  nodes of type~$\Index$ keep at most two children
  and nodes of type~$\Index'$ keep at most one child.
  (The behavior of the deleted children must be indistinguishable from the behavior of the remaining children,
  since otherwise an alarm would be raised.)
  This leaves us with a \ditree{}
  corresponding exactly to the input “expected” by the automaton $\Automaton_2$.
  Since it is accepted by $\Automaton_2$,
  this \ditree{} must be very close to an encoding
  of a solution $\Solution = \tuple{\Index_1, \dots, \Index_\SolutionSize}$,
  with the only difference that
  each element $\Index_\IIndex$ of~$\Solution$ may be represented
  by several nodes $\Node_\IIndex^1,\dots,\Node_\IIndex^\CloneCount$.
  However,
  we know by construction that $\Automaton$ behaves the same
  on all of these representatives.
  We can therefore remove the subtrees
  rooted at $\Node_\IIndex^2,\dots,\Node_\IIndex^\CloneCount$,
  and thus
  we obtain a \ditree{}-encoding of $\Solution$
  that is accepted by~$\Automaton$.
\end{proof}

\subsection*{Acknowledgments}
Fabian Reiter wants to thank Olivier Carton
for several pleasant discussions and constructive comments.
This work is supported by the
\href{http://delta.labri.fr/}{DeLTA project} (ANR-16-CE40-0007).

\bibliographystyle{eptcs}
\bibliography{da-paper.bib}

\end{document}